\documentclass{article}

\usepackage[varg]{txfonts}

\usepackage{graphicx} 
\usepackage{subfigure} 

\graphicspath{{./}{./figures/}}

\usepackage{natbib}

\usepackage{algorithm}
\usepackage{algorithmic}

\usepackage{amsmath}
\usepackage{amsfonts}
\usepackage{booktabs}
\usepackage{tabularx}

\newcommand{\eps}{\varepsilon}
\DeclareMathOperator{\cut}{cut}
\DeclareMathOperator{\vol}{vol}
\DeclareMathOperator{\Neigh}{Neigh}
\DeclareMathOperator{\flow}{\textit{flow}}

\usepackage[standard,thmmarks,amsmath]{ntheorem}

\newcommand{\alg}[1]{\textsc{#1}}

\usepackage[accepted]{icml2016} 
\usepackage{epstopdf}
\usepackage{url}
\usepackage{paralist}
\usepackage{stfloats}

\icmltitlerunning{A Simple and Strongly-Local Flow-Based Method for Cut Improvement}

\begin{document} 

\twocolumn[
\icmltitle{A Simple and Strongly-Local Flow-Based Method for Cut Improvement}

\icmlauthor{Nate Veldt}{lveldt@purdue.edu}
\icmladdress{Mathematics Department, Purdue University, West Lafayette, IN 47906}
\icmlauthor{David F. Gleich}{dgleich@purdue.edu}
\icmladdress{Computer Science, Purdue University, West Lafayette, IN 47906}
\icmlauthor{Michael W. Mahoney}{mmahoney@stat.berkeley.edu}
\icmladdress{International Computer Science Institute and Dept. of Statistics, University of California at Berkeley, Berkeley, CA 94720}
\icmladdress{}

\icmlkeywords{conductance, machine learning, ICML}

\vskip 0.3in
]

%
%


\begin{abstract} 
Many graph-based learning problems can be cast as finding a good set of vertices nearby a seed set, and a powerful methodology for these problems is based on maximum flows. We introduce and analyze a new method for locally-biased graph-based learning called \alg{SimpleLocal}, which finds good conductance cuts near a set of seed vertices. An important feature of our algorithm is that it is strongly-local, meaning it does not need to explore the entire graph to find cuts that are locally optimal. This method solves the same objective as existing strongly-local flow-based methods, but it enables a simple implementation. We also show how it achieves localization through an implicit $\ell_1$-norm penalty term. As a flow-based method, our algorithm exhibits several advantages in terms of cut optimality and accurate identification of target regions in a graph. We demonstrate the power of \alg{SimpleLocal} by solving problems on a 467 million edge graph based on an MRI scan.
\end{abstract}

\section{Introduction and Related Work}
\label{sec:intro2}

Finding good conductance cuts near a set of seed vertices in a graph is a well-studied and widely-applied problem in graph-based learning. Such an algorithm is strongly-local if its runtime depends only on the size of the seed set or on the size of the output set rather than on the size of the entire graph. Seeded PageRank and other spectral and random-walk methods give these guarantees. While these methods provide nice approximation guarantees, they fail to give optimal solutions and often exhibit ``sloppy'' boundaries when it comes to solving label propagation and community detection problems.

Flow-based methods exhibit numerous advantages including the ability to provide exact optimal solutions in some cases. The first strongly local flow-based method was introduced by \citet{Orecchia-2014-flow}, which exhibits a fast runtime but relies on a complicated variation of Dinic's algorithm for maximum flows, making it difficult to use in practice. In this paper we provide a new strongly-local algorithm that provides the same optimality guarantees while offering the flexibility of employing \emph{any} max-flow algorithm as a subroutine. While our algorithm's theoretical runtime is weaker than \citet{Orecchia-2014-flow}, we provide implementation details for it and demonstrate its ability to find low-conductance cuts in a large real-world dataset.

\paragraph{Graph-based learning.}
Graph-based learning is a recurring problem in machine learning where we are given a graph and some information about the nodes of this graph and the task is to infer the information on the unlabeled nodes. This is an instance of semi-supervised learning~\cite{Blum-2001-mincuts,Zhu-2003-diffusion} or transductive learning~\cite{Joachims-2003-transductive}. Algorithms for these problems on graphs are often called \emph{label propagation} methods due to their interpretation as spreading labels around a graph~\cite{Zhu-2003-diffusion,Fujiwara-2014-label-propagation}. Related problems include guided image segmentation~\cite{Mahoney-2012-local} and seeded community detection~\cite{andersen2006-communities,Kloumann-2014-seeds}, where we are given a set of sample pixels or nodes and the goal is to find the rest. 

Algorithms for graph-based learning largely split into three types: flow-based methods, spectral methods, and graph-based heuristics. Some of the seminal papers in semi-supervised learning on graphs and community detection discussed using minimum cuts in the network for this application~\cite{Blum-2001-mincuts,Flake-2000-communities}.  Subsequent papers found that \emph{spectral methods} had a number of advantages in terms of speed, unique solutions, and additional information about the \emph{strength} of the prediction~\cite{Zhu-2003-diffusion,Joachims-2003-transductive,Zhou-2003-semi-supervised}. Principled heuristic methods also abound~(e.g.~\citealt{Fujiwara-2014-label-propagation}) due to the simplicity of the setup.


Semi-supervised learning algorithms differ from typical graph algorithms in that they exhibit special locality properties. Typical graph algorithms, for example Kruskal's minimum spanning tree method, often depend on optimizing an objective function over the entire graph structure and return a result proportional to the size of the entire graph. In contrast, semi-supervised learning algorithms take as input an exogenously specified seed set of nodes, and return results that are biased towards a small part of the graph nearby the seed or reference set. If the running time is still dependent on the entire size of the graph, this is called \emph{weak locality.} For instance, solving the linear system involved in \citet{Zhou-2003-semi-supervised} returns a modestly sized set of nodes where the labels are expected to be located, but the linear system involves the entire graph.

For spectral methods, \citet{Spielman-2013-local} and \citet{andersen2006-local} have shown much stronger results. These algorithms are \emph{strongly-local}, in that the algorithm doesn't even access most of the nodes of the graph and thus the running time is dependent on the size of the seed set or output set, rather than the entire graph. Importantly, not only are these strongly-local spectral algorithms very fast (in both theory and practice, see~\citealt{Leskovec-2009-community-structure} and \citealt{Jeub-2015-locally}) but, when interpreted as graph partitioning methods, they come with locally-biased Cheeger-like quality of approximation guarantees with respect to the conductance objective. (Good conductance means small conductance; we define conductance later.)

In comparison, flow-based methods can be shown to \emph{optimally} solve the discrete partitioning objective, such as minimum conductance cut, if they are given an input that is not too large and that contains the desired set~\cite{Lang-2004-mqi} through a parameteric flow construction~\cite{Gallo-1989-parametric-flow}. A subsequent flow-based algorithm called \alg{Improve} had related exactness guarantees on the discrete objective but considerably relaxed the requirements on the input set~\cite{andersen2008-improve}. In the context of semi-supervised learning, the \alg{Improve} algorithm is a weakly-local method that would essentially find the optimal conductance set of nodes in the graph that is nearby the set of the seed labels. (We make a precise statement in Section~\ref{sec:improve}.) In fact, \alg{Improve} is essentially a flow-based analog of the spectral method used by \citet{Zhou-2003-semi-supervised} as shown by \citet{Gleich-2015-robustifying} using ideas from \citet{Mahoney-2012-local}. Recently, \citet{Orecchia-2014-flow} proposed a method called \alg{LocalImprove} that combined a slight modification to the discrete objective function in the \alg{Improve} algorithm with a variation on Dinic's algorithm for max-flow~\cite{Dinitz-1970-max-flow} in order to assemble the first flow-based method that is strongly-local. 
 

\paragraph{Strengths and weaknesses with spectral methods for graph-based learning and some fixes.}
In recent work, spectral methods have been found to have \emph{substantially} better theoretical guarantees, more akin to the guarantees of flow-based methods, if the resulting set of labels is \emph{very well-connected internally}~\cite{Zhu-2013-local}. However, when that doesn't hold, spectral methods tend to produce \emph{sloppy boundaries} unless the boundaries between labels is extremely clear. Figure~\ref{fig:sloppy} illustrates this effect on a simple synthetic construction. Given initial identical-labeled nodes, the spectral method~\cite{Zhou-2003-semi-supervised} diffuses over the boundary between groups too quickly and results in a potential misclassification. In contrast, the flow-based method (\alg{SimpleLocal}) is able to correctly identify the boundary given the same seeds. 

This example reflects a simplified scenario without the variety of fixes that are commonly used in spectral-based semi-supervised learning~\cite{Joachims-2003-transductive,Zhou-2011-error,Lu-2012-l1-ssl,Bridle-2013-p-Laplacians}. However, these spectral and strongly-local spectral methods have complicated theory with many parameters and options. This can make it difficult for non-experts to use, and it can be difficult to know what those fixes and strongly-local approximations  are optimizing exactly.  



\paragraph{Cut improvement.}
Cut improvement is a problem framework where we are given an initial partitioning of a graph into two pieces and the goal is to identify a \emph{better} split according to some quotient of cut and size. Both conductance and its relative, \emph{quotient cut}, fit into this general framework. Algorithms here date back to initial work on parametric maximum flow~\cite{Gallo-1989-parametric-flow}. There are a variety of methods that use the submodular property of the objective and various decompositions of submodular functions to solve them~\cite{Patkar-2003-improving,Narasimhan-2007-local-search}. The original flow-based methods \alg{MQI} and \alg{Improve} had attractive theoretical guarantees and empirical performance for this task~\cite{Lang-2004-mqi,andersen2008-improve}.
The tendency of spectral methods to make errors around the boundary was also recognized in this literature~\cite{Lang2005-spectral-weaknesses} and these flow-based methods were already a well-known fix. 
Along these lines, \citet{Chung-2007-local-cuts} provided a spectral analogue of \alg{MQI} and \citet{Mahoney-2012-local} provided \alg{MOV}, a weakly-local version of spectral graph partitioning, which may be interpreted as a spectral analogue of \alg{Improve}.
In fact, \citet{Gleich-2015-robustifying} showed that the \alg{MOV} method is optimizing the same solution as the semi-supervised method of \citet{Zhou-2003-semi-supervised}.

\paragraph{Summary of Contributions.}
In this paper, we first show that the modification to the \alg{Improve} objective used in the strongly-local \alg{LocalImprove} can be precisely stated as an $\ell_1$-penalized \alg{Improve} objective (Theorem~\ref{thm:l1}).  This makes precise the sense in which \alg{LocalImprove} \emph{implicitly} optimizes a sparsity-induced regularized version of \alg{Improve}. The authors of \alg{LocalImprove} made use of a sophisticated white-box modification of Dinic's algorithm to provide the best possible runtime. Our main contribution is then a new strongly-local flow algorithm that uses existing max-flow algorithms as a black box (Algorithm~\ref{alg:3SF}, Theorem~\ref{thm:maxflow}). Our algorithm solves the same optimization problem as \alg{LocalImprove} and our aim is to provide an algorithm that is both flexible and easy to implement while still being strongly-local. Thus we call our algorithm \alg{SimpleLocal}.

In combination, these two results deconvolve the origin of the local solution, which occurs because of the the $\ell_1$-regularization applied to the problem, from the algorithm that identifies this local solution. The second result enables us to create an extremely simple and scalable implementation of \alg{SimpleLocal} where \emph{any} max-flow algorithm can be used to solve this sequence of problems including efficient GPU methods~\cite{He-2010-gpu-max-flow} or the latest exact theoretical algorithms~\cite{Orlin-2013-max-flow}. We use this implementation to show a few examples of how this new method is able to solve problems on graphs originating from MRI data with 467 million edges in a few minutes.


\begin{figure}[t]
	\centering
	\subfigure[Graph-based learning with a spectral method]{\includegraphics[trim={0cm 6cm 6cm 0},clip,width=0.45\linewidth]{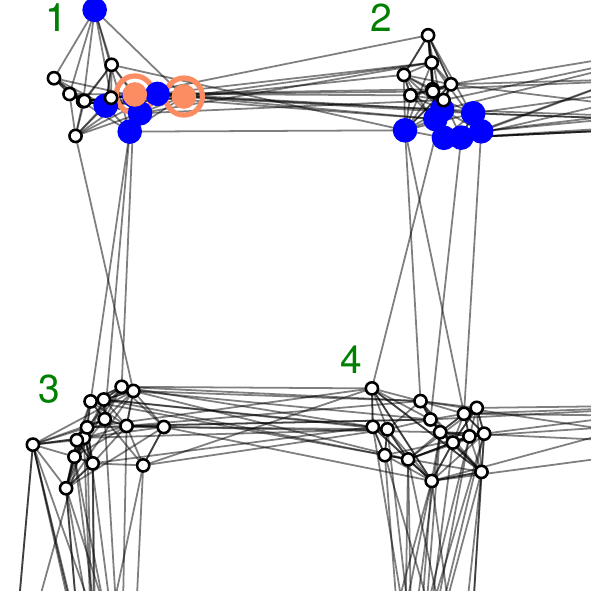}}%
	\hfill
	\subfigure[Graph-based learning with a local flow method]{\includegraphics[trim={0cm 6cm 6cm 0},clip,width=0.45\linewidth]{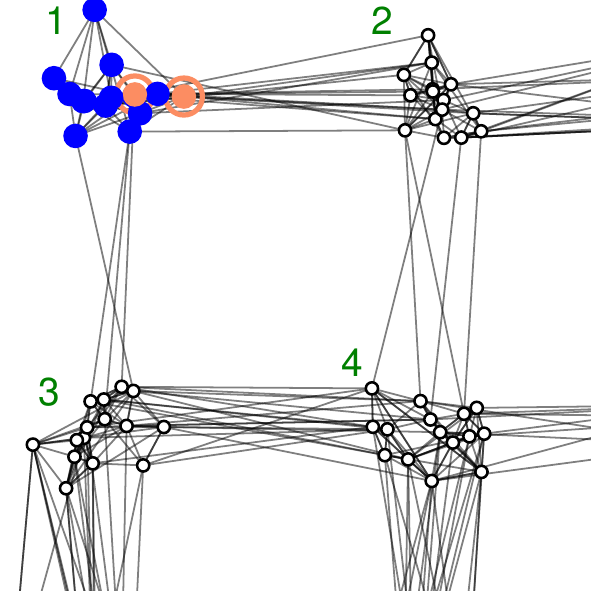}}
	\caption{There are four synthetic labels in displayed region of the graph, one for each group. The vertices are connected based on nearest neighbors and sparsified longer range connections. Using the two seed nodes marked in orange, the spectral prediction, the blue nodes in (a), move to the adjacent group whereas the flow predictions, the blue nodes in (b), accurately capture the true boundary.}
	\label{fig:sloppy}
\end{figure}

%

%



\section{Preliminaries and Notation}
Let $G = (V,E)$ be an undirected, unweighted graph with $n = |V|$ nodes and $m = |E|$ edges. For a given vertex $v \in V$, the degree $d_v$ is equal to the sum of edges incident to $v$, and the volume of a subset of nodes $S \subset V$ is defined to be $\vol(S) = \textstyle \sum_{s \in S} d_s.$
Given two subsets of vertices $A$ and $B$, we indicate the set of edges between them by \[\cut(A,B) = \cut(B,A) = \{(i,j) \in E: i \in A, j \in B\}.\]
We associate every vertex set $S \subset V$ with the set of edges between $S$ and the rest of the graph, $\cut(S) = \cut(S, \bar{S})$, where $\bar{S} = V\backslash S$.  We use $\partial S= |\cut(S,\bar{S})|$ to indicate the number of edges in this set.
Let $\Neigh(S)$ indicate the nodes that are not included in $S$ but share an edge with $S$,
\[ \Neigh(S) = \{v \in \bar{S}: (v, s) \in E \text{ for some } s \in S\}. \]

We measure how well-connected the set $S$ is by its conductance $\phi(S)$, defined by
\[ \phi(S) = \tfrac{\partial S}{\min\{\vol(S), \vol(\bar{S})\}}. \]

\section{Implicit Sparsity Regularization}

In this section we present a new result which relates the objectives of \alg{Improve} and \alg{LocalImprove}. We begin by reviewing the construction of the \emph{augmented graph} of \citet{andersen2008-improve} used in \alg{Improve}, and the modification of this graph introduced by~\citet{Orecchia-2014-flow}. We relate both of these graphs back to our original problem by showing how low-capacity cuts in the augmented and modified augmented graphs correspond to low-conductance cuts in the original input graph. Our main result in this section is to show that the min-cut objective solved by \alg{LocalImprove} is implicitly equivalent to a sparsity-inducing $\ell_1$-regularization of the min-cut objective solved by \alg{Improve}. 
This result guides our understanding of strongly-localized flow-based cut improvement methods, and sheds light on the success and robustness of algorithms such as \alg{LocalImprove} and \alg{SimpleLocal}. This is also an example of algorithmic anti-differentiation~\cite{Gleich-2014-alg-anti-diff} where we characterize the optimization problem that \alg{LocalImprove} was implicitly solving as a result of their algorithmic setup. 


\subsection{\alg{Improve} and the Augmented Graph}
\label{sec:improve}

We begin with a graph $G = (V,E)$, an initial seed set $R\subset V$ satisfying $\vol(R) \leq \vol(\bar{R})$, and a parameter $\alpha \in (0,1)$. The \emph{augmented graph} $G_R(\alpha)$ is constructed through the following steps:

\begin{compactenum}
\item Retain original nodes, edges, and edge weights of $G$
\item Add a source node $s$ and a sink node $t$ 
\item For every $r\in R,$ add an edge $(s,r)$ with capacity $\alpha d_r$ 
\item For every $v \in \bar{R}$, add edge $(v,t)$ with capacity $\alpha f(R) d_v.$
\end{compactenum}
Here $f(R) = \frac{\vol(R)}{\vol(\bar{R})}$ is chosen so that the total capacity into the sink equals the total capacity out of the source. See Figure~\ref{fig:auggraph} for a visualization of a small augmented graph.

\begin{figure}[tbp] 
   \centering
   \includegraphics[width=1.7in]{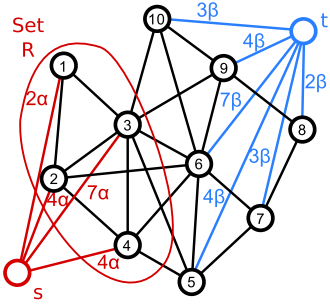} 
   \caption{An illustration of the augmented graph $G_R(\alpha)$  used by \alg{Improve}, where $\beta = \alpha f(R)$.  If we change the sink-side weight so that $\beta = \alpha(f(R) + \delta)$, this corresponds to the modified augmented graph $G'_R(\alpha,\delta)$ used by \alg{LocalImprove} and our \alg{SimpleLocal}. }
   \label{fig:auggraph}
\end{figure}
An $s$-$t$ cut of the augmented graph is any set of edges which, when cut, partitions the nodes into disjoint sets where the source and sink are in separate sets. Any $s$-$t$ cut can be associated with the set of nodes $S \subset V$ from the original graph  that are on the same side of the cut as $s$. The capacity of any $s$-$t$ cut of $G_R(\alpha)$ is equal to 
\begin{equation} \label{eq:objfn}
\alpha \vol(R)+  \partial S - \alpha \vol(R\cap S) + \alpha f(R) \vol(\bar{R}\cap S),
\end{equation}
where $S$ is the set of edges on the source side. 

The \alg{Improve} algorithm solves a sequence of minimum $s$-$t$ cut problems on $G_R(\alpha)$ for decreasing values of $\alpha$.  This is exactly equivalent to solving a sequence of optimization problems where we minimize the objective \eqref{eq:objfn} over sets $S \subset V.$ The goal is to find the smallest $\alpha$ such that the minimum $s$-$t$ cut (i.e. the optimal $S$) is less than $\alpha \vol(R)$. Note that regardless of $\alpha$ we can always achieve a cut with capacity $\alpha \vol(R)$ by selecting all edges from the source to the rest of the graph. 
This type of problem is an instance of a parametric max-flow~\cite{Gallo-1989-parametric-flow}; although those techniques are unnecessary for \alg{Improve}.
 

\subsection{LocalImprove and the Modified Augmented Graph}

The \emph{modified augmented graph} $G'_R(\alpha,\delta)$ is obtained by increasing the weight of the edges from $\bar{R}$ to $t$ to $\alpha \eps d_w$ for all nodes $w\in \bar{R},$ where $\eps = f(R) + \delta$ for $\delta \geq 0.$ Finding the minimum $s$-$t$ cut of $G'_R(\alpha,\delta)$ is equivalent to minimizing a slightly modified objective:
\begin{equation}
 \begin{aligned} & \alpha \vol(R)+  \partial S - \alpha \vol(R\cap S) + \alpha f(R) \vol(\bar{R} \cap S) \\ & \qquad + \alpha \delta \vol(\bar{R} \cap S). \end{aligned}
\end{equation}

By including the extra term $\alpha \delta \vol(\bar{R}\cap S)$, this in effect increases the penalty of including nodes outside of $R$. Note that $G_R(\alpha) = G'_R(\alpha,0).$ 

Cuts in both $G_R(\alpha)$ and $G'_R(\alpha,\delta)$ are easily associated with sets of vertices in the original graph $G$. Given a max $s$-$t$ flow on either graph, the set of nodes reachable from $s$ via unsaturated edges (excluding $s$ itself) forms a subset $S \subset V$ in the original graph $G$.

\alg{LocalImprove} finds a set $S$ with low conductance by solving a sequence of \emph{approximate} max flow computations on $G'_R(\alpha,\delta)$ for different values of $\alpha$. The method relies on modifying Dinic's max-flow algorithm and a procedure for finding blocking flows on a subgraph of $G'_R(\alpha,\delta)$ called the \emph{local graph}. The local graph is updated and expanded as needed at each step to allow more flow to be routed from $s$ to $t$. For more details we direct the reader to~\citet{Orecchia-2014-flow}.

\subsection{Relating Cuts in Modified Augmented Graph to Low-Conductance Sets in the Original Graph}

The following lemma is a generalization of a core result of~\citet{andersen2008-improve} and is proven as a part of Lemma 3.1 in~\citet{Orecchia-2014-flow}.

\begin{lemma}
 If the minimum $s$-$t$ cut of $G'_R(\alpha,\delta)$ for $\delta \geq 0$ is less than $\alpha \vol(R),$ then $\phi(S) < \alpha$, where $S$ is the node set corresponding to the cut.
\end{lemma}
(For a proof see the supplementary material.)

The task of \alg{Improve} is to find the smallest $\alpha$ such that the maximum $s$-$t$ flow of the augmented graph is less than $\alpha \vol(R).$ This is equivalent to minimizing the function 
\begin{equation} \label{eq:relquot}
\phi_R(S) = \partial S/(\vol(R\cap S) - f(R) \vol(\bar{R} \cap S))
\end{equation}
which~\citet{andersen2008-improve} refer to as the \emph{quotient score of} $S$ \emph{relative to the seed set} $R$. Both \alg{LocalImprove} and our new algorithm \alg{SimpleLocal} minimize a similar quotient function, 
given later as equation \eqref{eq:modrelcond}.
\subsection{Relating $s$-$t$ Cuts of the Augmented Graph and Modified Augmented Graph}

Our first new result is to show that finding a minimum $s$-$t$ cut of the modified augmented graph is equivalent to solving an $\ell_1$ sparsity-regularized version of the min-cut objective for a related, standard augmented graph. This result gives insight into why \alg{LocalImprove} succeeds in finding a cut with similar conductance guarantees to those provided by \alg{Improve}, despite only exploring a portion of the graph. This result is analogous to a similar relationship discovered in \citet{Gleich-2014-alg-anti-diff} between the Andersen, Chung, Lang procedure and an $\ell_1$-regularized $\ell_2$-version of the min-cut objective on the \alg{Improve} augmented graph~\cite{andersen2006-local}.

\begin{theorem} \label{thm:l1}
	Finding the minimum cut of the modified augmented graph $G'_R(\alpha,\delta)$ is equivalent to solving an $\ell_1$-regularized version of the minimum $s$-$t$ cut objective of a related augmented graph $G_R(\beta)$. More specifically, the set $S$ which minimizes the $s$-$t$ cut objective of $G_R'(\alpha,\delta)$:
	\begin{equation}
	\alpha \vol(R)+  \partial S - \alpha \vol(R\cap S) + \left( \alpha f(R) + \alpha \delta\right)\vol(\bar{R} \cap S)
	\label{eq:augmin}
	\end{equation}
	also minimizes the $\ell_1$-regularized objective of $G_R(\beta)$:
	\begin{equation}
	\label{eq:reg}
	\beta \vol(R) +  \partial S - \beta \vol(R\cap S) + \beta f(R) \vol(\bar{R}\cap S) + \kappa \vol(S),
	\end{equation}
	where $\kappa = \frac{\alpha \delta}{1 + f(R)} > 0$ and $\beta = \alpha + \kappa.$ 
\end{theorem}

\begin{proof}
If we note that $\vol(\bar{R}\cap S) + \vol(R\cap S)= \vol(S)$ and substitute $\alpha \delta = \kappa + \kappa f(R)$ and $\alpha = \beta - \kappa$, we get
\begin{align*}
& -\alpha \vol(R\cap S) + (\alpha f(R) + \alpha \delta) \vol(\bar{R}\cap S)\\ 
& \qquad = -(\beta - \kappa) \vol(R\cap S) + (\beta f(R) + \kappa) \vol(\bar{R} \cap S)\\
&\qquad = -\beta \vol(R\cap S) + \beta f(R) \vol(\bar{R}\cap S) + \kappa \vol(S).
\end{align*}
The constant term in each objective does not affect the optimal set $S$,  so we see that~\eqref{eq:augmin} and~\eqref{eq:reg} are equivalent.
\end{proof}
\paragraph{Remark.} The additional term in the regularized objective \eqref{eq:reg} is $\kappa \vol(S)$. When this objective is converted into a linear program for min-cut, it is: 
\begin{equation}
 \begin{array}{ll} \min_{x} & \sum_{(i,j) \in E_R(\beta)} c_{i,j}(\beta) |x_i - x_j| + \kappa \sum_{i} | d_i x_i| \\ \text{s.t.} & 0 \le x_i \le 1, x_s = 1, x_t = 0, \end{array}
\end{equation}
where $c_{i,j}(\beta)$ is the capacity of each edge in the augmented graph with $\beta$ from the theorem and where $d_i$ is the degree of node $i$ in the original graph.
The extra term is then exactly an $\ell_1$ penalty. 

\section{SimpleLocal Algorithm}
Our primary contribution is the algorithm \alg{SimpleLocal}, a simplified framework for computing the objective of \alg{LocalImprove}.
Just as \alg{LocalImprove}, we rely on constructing and updating a local subgraph of $G'_R(\alpha,\delta)$. However, rather than using Dinic's algorithm to compute approximate maximum flows, we develop a new three-stage method for exact maximum flow computations on $G'_R(\alpha,\delta)$. 

%
\subsection{Three-Stage Local Max Flow Procedure}
We begin with a detailed explanation of \alg{3StageFlow}, the newly designed algorithm we employ to compute a maximum $s$-$t$ flow of a given modified augmented graph $G_R'(\alpha,\delta)$. After constructing an initial local graph, our algorithm enters a three-stage process that is repeated until convergence to a maximum flow. In each iteration we expand the local graph, compute a small-scale maximum $s$-$t$ flow, and then update the local graph based on this flow. By iteratively growing the local graph and increasing our small-scale flow computations on it in this way, we will converge to an $s$-$t$ flow that is a maximum on all of $G'_R(\alpha,\delta)$. 

\textbf{Initialization.}
Let $G'$ denote the modified augmented graph $G'_R(\alpha,\delta)$. We begin by forming the local graph $L = (V_L,E_L)$, a subgraph of $G'$ which includes:
\begin{compactitem}
\item Nodes $\{s,t\}\cup R \cup \Neigh(R)$
\item Edges from $s$ to $R$
\item Edges between distinct nodes in $R$
\item Edges from $R$ to $\Neigh(R)$
\item Edges from $t$ to $\Neigh(R).$
\end{compactitem}
Let $F$ denote our flow vector, and $\flow(F)$ indicate the total amount of flow routed from $s$ to $t$. Initially $F$ is set to the zero vector.

\textbf{Stage 1. Expansion.} At the beginning of each new iteration we expand the local graph to allow more flow to be routed from $s$ to $t$. We use $X$ to denote the set of nodes to expand on at the beginning of an iteration. For any node $x \in X$, we add all neighbors of $x \in G'$ that are not yet a part of $L$, and also include all edges from $x$ to all its neighbors. For each new node added to $L$, we include the edge it shares with the sink $t$. In the first step we have no need to expand the local graph yet, so we set $X = \emptyset$.

\textbf{Stage 2. Max-Flow Computation.} Once $L$ is correctly expanded, we compute the maximum flow $f$ on the local graph $L$ using any available max-flow subroutine. We then update our flow vector $F \leftarrow F + f$.

Let $L_f$ denote the residual graph of the flow. This graph is formed by replacing the capacity $c_{ij}$ of an edge in $E_L$ by $c_{ij} - f_{ij}$, where $f_{ij}$ is the flow on edge $(i,j)$, and where the capacity of edge $(j,i)$ is replaced with the value $f_{ij}$. 

\textbf{Stage 3. Updates.} After computing a maximum flow, we resolve the effects of the flow and determine whether the local graph should be further expanded. We begin by updating the local graph to be the residual graph of $f$, and find the set of nodes still connected to $s$ by a chain of unsaturated edges. We refer to this as the \emph{source set} $S$. When we converge to a max flow, this is the set that is returned.

We determine the set of nodes around which to expand $L$ in the same way \alg{LocalImprove} does after computing a localized blocking flow~\cite{Orecchia-2014-flow}. The nodes to expand on are exactly those whose edge to $t$ was saturated by the flow $f$. An edge $(v,t)$ is saturated when the flow $f_{vt}$ is equal to the available capacity from node $v$ to $t$, so we determine the new expansion set $X$ as follows:
\[X \leftarrow \{v \in V_L: f_{vt} = c_{vt} \text{ for $f$} \}.\]
If $X$ is non-empty, there exists at least one node $x \in X$ in the source set that has edges and neighboring nodes not yet included in $L$. This implies more flow could be routed from $s$ to $t$ through $x$. If $X$ is empty, we will show in the next section that the current flow $F$ is optimal and we no longer need to expand the local graph. 

An outline for \alg{3StageFlow} is given in Algorithm \ref{alg:3SF}. 
\begin{algorithm}[tb]
   \caption{\alg{3StageFlow}}
   \label{alg:3SF}
\begin{algorithmic}
   \STATE {\bfseries Input:} graph $G$, parameters $\alpha, \delta$, seed set $R$
   \STATE {\bfseries Initialize:} 
   \STATE $V_L:= \{R,\Neigh(R), s,t\}$
   \STATE $E_L:= \{(s, R), (R,\Neigh(R)), (\Neigh(R),t) \}$
   \STATE $F:=0$; $\quad$ $X := \emptyset$
   \WHILE{$X \neq \emptyset$ \text{ or } $F = 0$}
   \STATE \textbf{1. Expand $W$}
   \FOR{$x \in X$}
   \STATE $V_L \leftarrow V_L \cup \Neigh(x)$
   \STATE $E_L \leftarrow E_L \cup \{(x, v): v \in V_L\} \cup\{(y, t): y \in \Neigh(x)\}$
      \ENDFOR
   \STATE \textbf{2. Max Flow}: 
   \STATE$f \leftarrow \alg{MaxSTflow}(L)$; $\quad$ $F \leftarrow F + f$
   \STATE \textbf{3. Update $L$}
   \STATE $L\leftarrow L_f$; $\quad$ $S\leftarrow \emph{ source set}$
   \STATE $X \leftarrow \emph{nodes whose edge to $t$ was saturated}$
   \ENDWHILE
\end{algorithmic}
\end{algorithm}

\subsection{Convergence of \alg{3StageFlow}}
The following lemma is analogous to a result shown for the \alg{LocalImprove} algorithm~\cite{Orecchia-2014-flow}. We use it to prove that \alg{3StageFlow} converges to a maximum $s$-$t$ flow of $G'$.
\begin{lemma}
If $S$ is the set of nodes returned by \alg{3StageFlow}, then 
\[ S \subseteq R \cup P_x,\]
where $P_x$ is the set of nodes we have previously expanded on.
\end{lemma}

\begin{proof} The algorithm terminates when $X = \emptyset$ and $F > 0$. If we assume $S$ is not a subset of $R \cup P_x$, then there exists a node $x \in S$ such that $x \notin R$ and $x \notin P_x$. Because $x \in \bar{R}$, this node must share an edge in the local graph with $t$. Since $x \in S$, there is a path of unsaturated edges connecting $s$ and $x$, so in order for $f$ to be maximal the edge $(x,t)$ must be saturated. This is a contradiction, because $X = \emptyset$ implies that edge $(x,t)$ was not saturated on the most recent iteration, and $x \notin P_x$ implies we did not previously expand on $x$, meaning $(x,t)$ was not saturated in any previous iteration.
\end{proof}

We can now prove the optimality of the set returned by \alg{3StageFlow}.

\begin{theorem} \label{thm:maxflow} When  $X = \emptyset$ and $F>0$, $F$ is a maximum flow of $G'$ and $\cut(S\cup\{s\})$ is the minimum $s$-$t$ cut set of $G'$.
\end{theorem}
\begin{proof}
We include the requirement $F > 0$ to indicate we have not stopped before the first iteration. 
In the local graph, the set of saturated edges between $S$ and $V_L\backslash S$ defines an $s$-$t$ cut with capacity equal to the total amount of flow routed from $s$ to $t$. The capacity of any $s$-$t$ cut is an upper bound for the amount of flow that can be routed from source to sink, so we see that $F$ and $\cut(S\cup\{s\})$ are optimal in $L$.
By the above lemma, $S \subseteq R\cup P_x$. This implies that all neighbors and edges of $S$ in $G'$ are already included in the local graph $L$. Therefore, the max-flow and min-cut of the local graph is also an optimal flow and cut pair for the entire graph $G'.$
\end{proof}

\subsection{Strong-locality and Runtime Guarantee}
The explored portion of $G'$ directly corresponds to a subgraph of $G$ in the following way: if we consider the local graph $L$ and remove $s$ and $t$ and all edges incident to them, we end up with a subgraph of $G$ which we call the \emph{explored subgraph} and denote $G_{\text{exp}}$. This explored region is exactly the subgraph of $G$ that our method would need to extract to create the small max-flow problems. Our algorithm not only obtains a maximum $s$-$t$ flow on the entire graph, but we can show that it does so without exploring the entire graph.  This result substantially sharpens a related result from \citet[Theorem 1a]{Orecchia-2014-flow}.
\begin{theorem}
 Given a graph $G$, seed set $R$, and locality parameter $\delta > 0$, the \alg{3StageFlow} procedure explores a subgraph of $G$ satisfying the following bound:
\[\vol(G_{\text{exp}}) \leq \vol(R)\left( 1 + \frac{2}{\eps} \right) + \partial R, \]
where $\eps = f(R) + \delta.$
\end{theorem}

\begin{proof} We first bound the expanded set $P_x$. For any $\alpha$, the maximum flow of $G'$ is bounded above by $\alpha \vol(R)$, the capacity of the edges leading out of $s$. If we expand on a node $v \in L$, it means in the previous iteration the edge $(v,t) \in E_L$ was saturated, implying that $\alpha \eps d_v$ flow was routed to $t$. The total amount of flow through these expanded nodes therefore must satisfy
\[\textstyle \sum_{p \in P_x} \alpha \eps d_p \leq \alpha \vol(R),\]
which gives the bound
\[\vol(P_x) = \textstyle \sum_{p \in P_x} d_p \leq \tfrac{\vol(R)}{\eps}. \]
By the construction and update procedure of $L$, the vertex set of $G_{\text{exp}}$ is $R\cup P_x \cup Q$, where $Q = \Neigh(R\cup P_x)$. This subgraph includes all edges incident to nodes in $R\cup P_x$, but contains no edges between nodes in $Q$, the nodes around which $L$ has not been expanded. Because of this, the volume of $Q$ can be upper bounded by the volume of $P_x$ and the cut of $R$ as follows:
\[ \vol(Q) =  |\cut(Q,P_x)| + |\cut(Q,R)| \leq \vol(P_x) + \partial R. \]

This can be used to show the final result:
\[ \begin{aligned}
\vol(G_{\text{exp}}) &= \vol(R) +\vol(P_x) + \vol(Q) \\ &\leq \vol(R) + 2\vol(P_x) + \partial R \leq \vol(R)(1 + \tfrac{2}{\eps} ) + \partial R.
\end{aligned} \]
\end{proof}

This result implies the following---extremely crude---strongly-local runtime guarantee. In the worst case, each flow-problem takes $O(\vol(R)^2/\eps^2)$ to solve using the algorithm from \citet{Orlin-2013-max-flow}. We have at most one flow problem for each edge in the final local graph (if we grew the local graph by one vertex at a time, we must get at least one edge), giving an overall strongly-local bound of $O(\vol(R)^3/\eps^3)$. This is highly conservative and worse than the bound on \alg{LocalImprove} from \citet{Orecchia-2014-flow}; we expect real-world runtimes to be substantially faster. 

\subsection{Full Outline of SimpleLocal}

Given a graph $G$ with reference set $R$, \alg{SimpleLocal} finds a good conductance cut 
by repeatedly calling \alg{3StageFlow} to find the smallest $\alpha$ such that the maximum $s$-$t$ flow of $G'_R(\alpha,\delta)$ is less than $\alpha \vol(R).$ 
\begin{algorithm}[H]
\caption{\alg{SimpleLocal}}
\label{alg:SL}
\begin{algorithmic}
\STATE \textbf{Input:} $G$, $R$, locality parameter $\delta \geq 0$
\STATE $\alpha := \phi(R)$
\STATE $[F,S] := \alg{3StageFlow}(G'_R(\alpha,\delta))$
\WHILE{$flow(F) < \alpha \vol(R)$}
\STATE $\alpha \leftarrow \phi(S)$; $\quad$ $S^* \leftarrow S$
\STATE $[F,S] := \alg{3StageFlow}(G'_R(\alpha,\delta))$
\ENDWHILE
\STATE \textbf{Return:} $S^*$
\end{algorithmic}
\end{algorithm}
This procedure finds the set $S^*$ that minimizes
\begin{equation} 
\label{eq:modrelcond}
\bar{\phi}_R(S) = \partial S/(\vol(R\cap S) - \eps \vol(\bar{R} \cap S)), 
\end{equation}
which is related to the relative quotient score~\eqref{eq:relquot}.

\subsection{Cut Quality Guarantee}
The following result is an extension of the theorem from~\citet{andersen2008-improve}, updated to include the effects of our parameter $\delta.$ 
\begin{theorem}
Given an initial reference set $R \subset V$ with $\vol(R) \leq \vol(\bar{R})$, \alg{SimpleLocal} returns a cut set $S^*$ where
\begin{enumerate}
\item if $C\subseteq R$, then $\phi(S^*) \leq \phi(C).$
\item For all sets of nodes $C$ such that for some $\gamma > \delta$
\[\tfrac{\vol(R\cap C)}{\vol(C)} \geq \tfrac{\vol(R)}{\vol(V)} + \gamma \tfrac{\vol(\bar{R})}{\vol(V)},\]
we have
$\phi(S^*) \leq \frac{1}{(\gamma - \delta)} \phi(C).$
\end{enumerate}
\end{theorem}
(We include a full proof in the supplementary material.)

\section{Experiments}
In this section we present experimental results for \alg{SimpleLocal} on two graphs. We begin with an example on a small collaboration network to illustrate the effect of the locality parameter $\delta$. We then turn our attention to graphs from MRI scans to demonstrate \alg{SimpleLocal}'s ability to solve problems on extremely large graphs. Our implementation of \alg{SimpleLocal} and \alg{3StageFlow} are in Matlab, using Gurobi to solve the max-flow problems.

\subsection{Netscience Example}
Newman's netscience graph is a collaboration network with 379 nodes and a total volume of 1828. 
The reference set we use is a node and its immediate neighbors. We run \alg{SimpleLocal} for decreasing values of $\delta$ from 1 to 0 (and implicitly, decreasing amounts of regularization) to obtain cuts near $R$ of increasing size. These sets have increasingly better conductance. Note that for $\delta = 0$ we are computing the \alg{Improve} objective. We illustrate our results in Figure~\ref{fig:Netscience}.
\begin{figure}[tb]
\centering
\subfigure[$\delta = 1$, $\phi = .47$ \newline volume explored 94 \newline (bound gave 211)]{%
\includegraphics[width=1.2in]{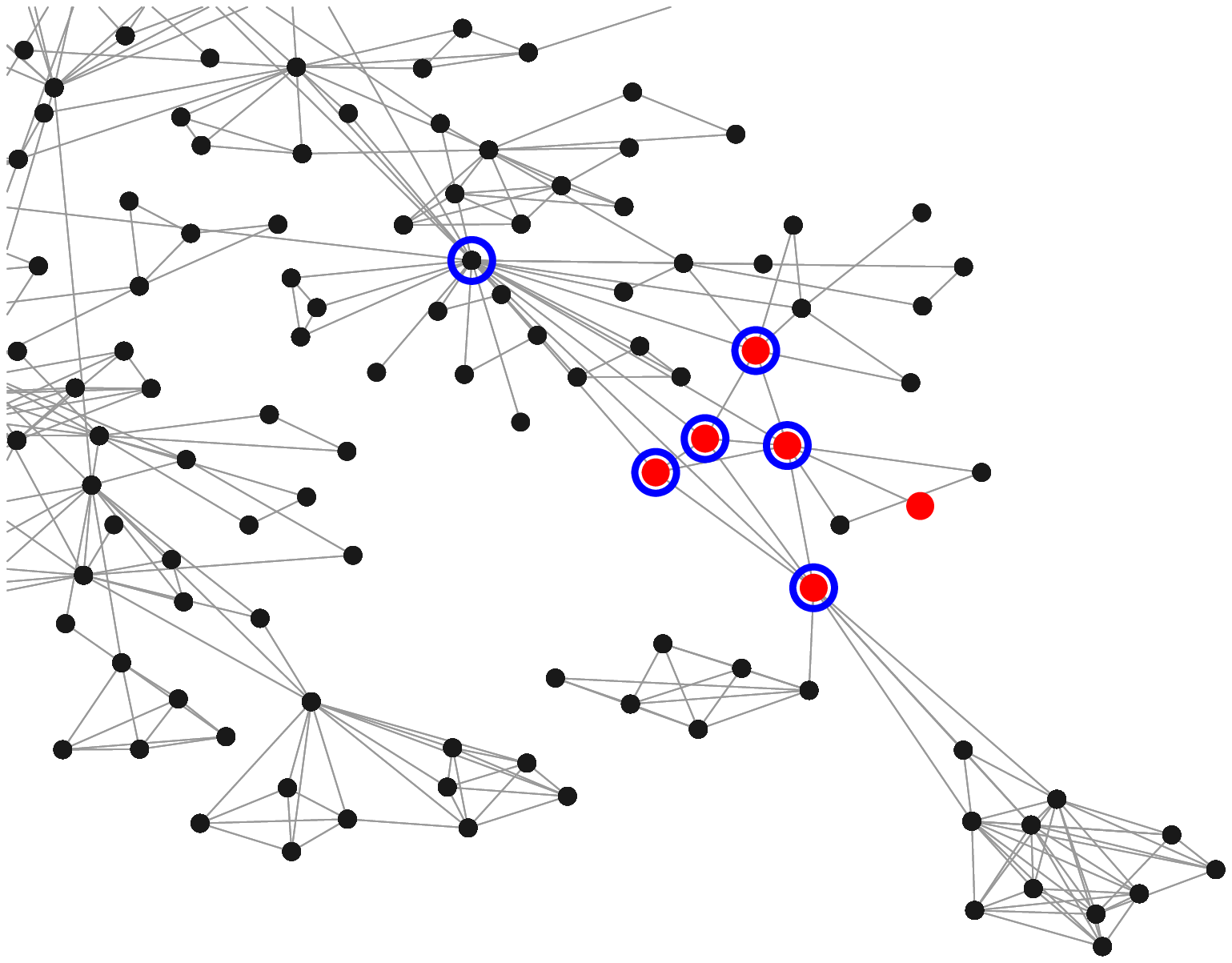}
\label{fig:subfigure1}} $\qquad$
\subfigure[$\delta = .6$, $\phi = .23$ \newline volume explored 116 \newline (bound gave 284)]{%
\includegraphics[width=1.2in]{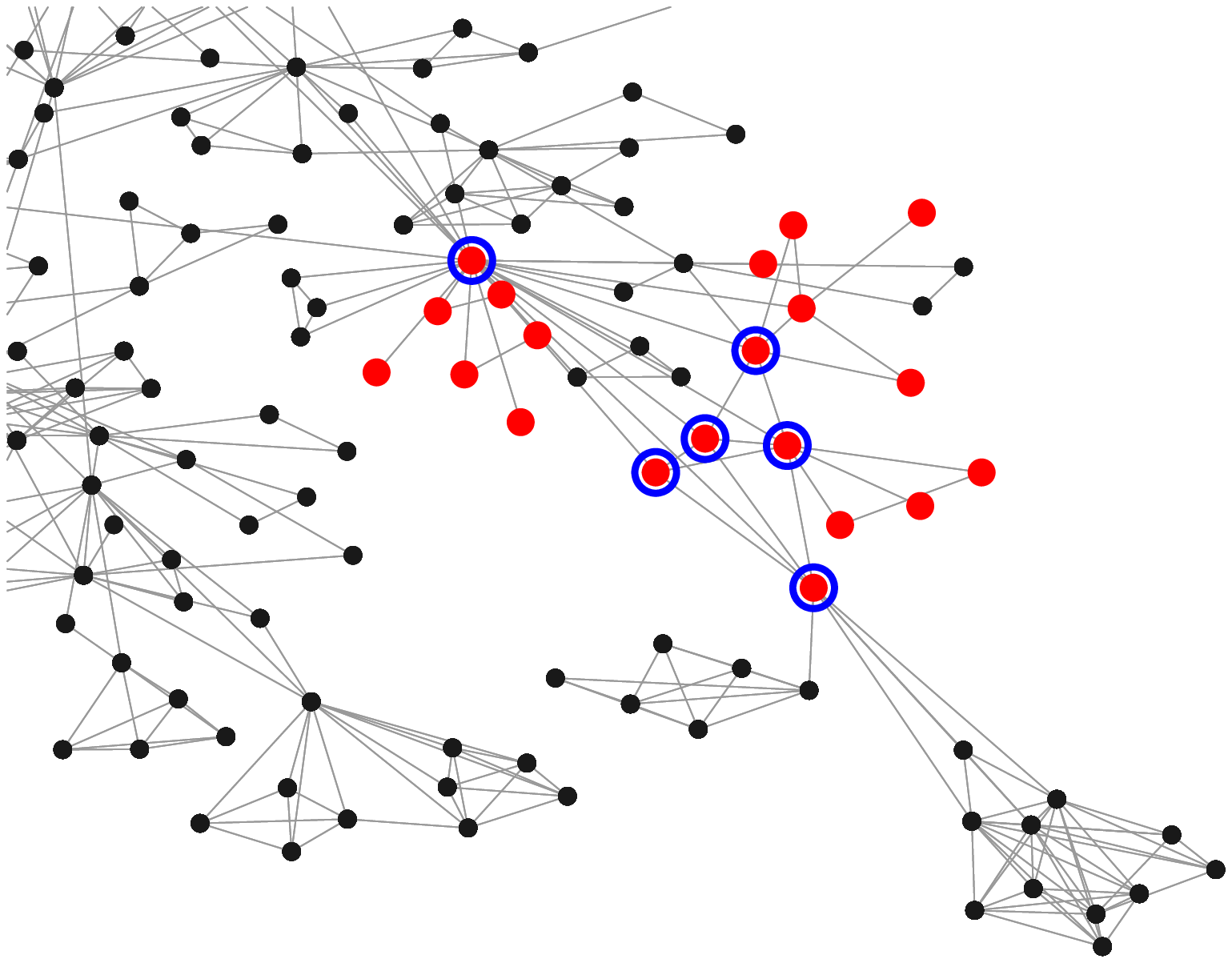}
\label{fig:subfigure3}}
\subfigure[$\delta = .3$, $\phi = .09$ \newline volume explored 160  \newline (bound gave 455)]{%
\includegraphics[width=1.2in]{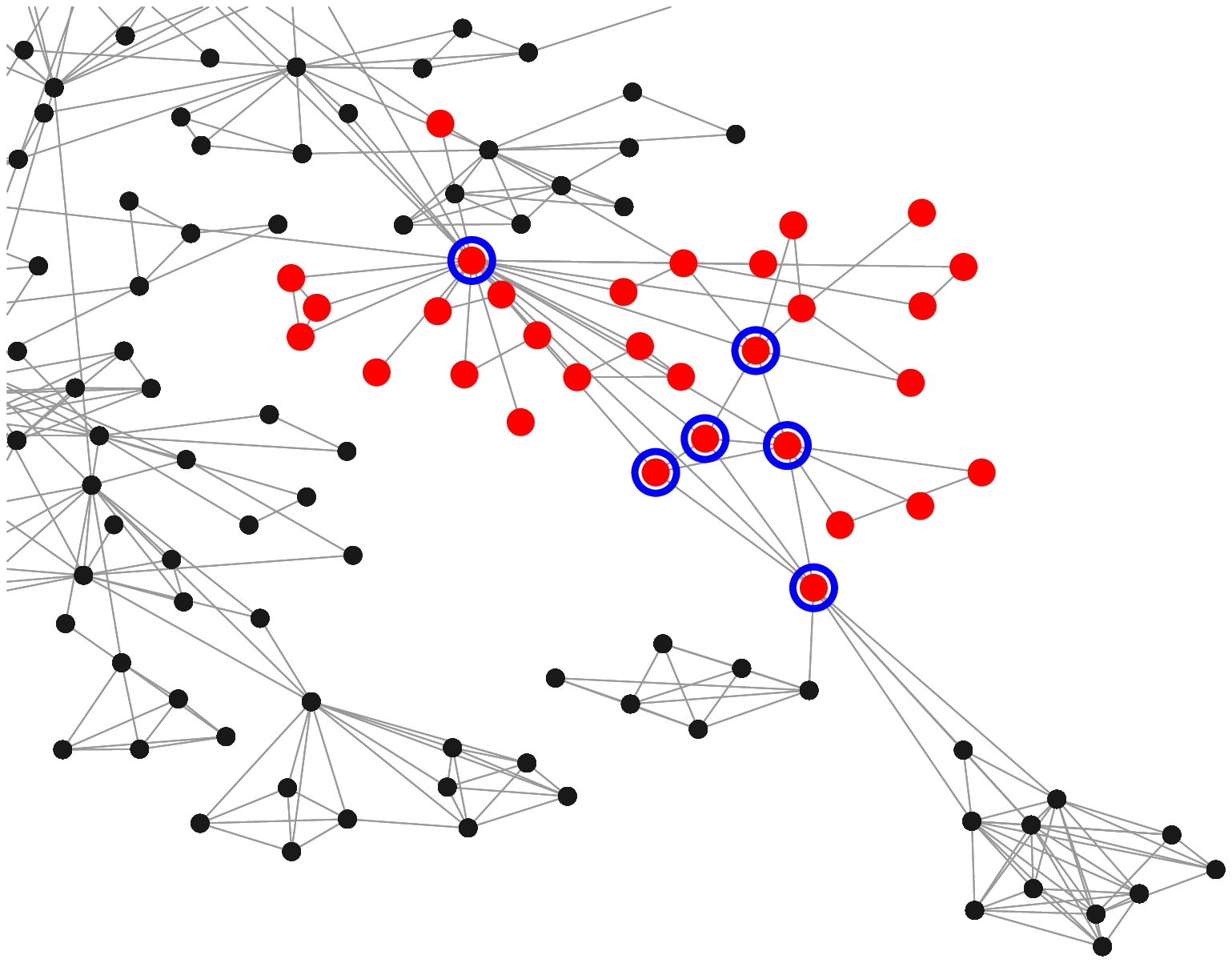}
\label{fig:subfigure5}}
 $\qquad$
\subfigure[$\delta = 0$, $\phi = .03$ \newline volume explored 522 \newline (bound gave $\infty$)]{%
\includegraphics[width=1.2in]{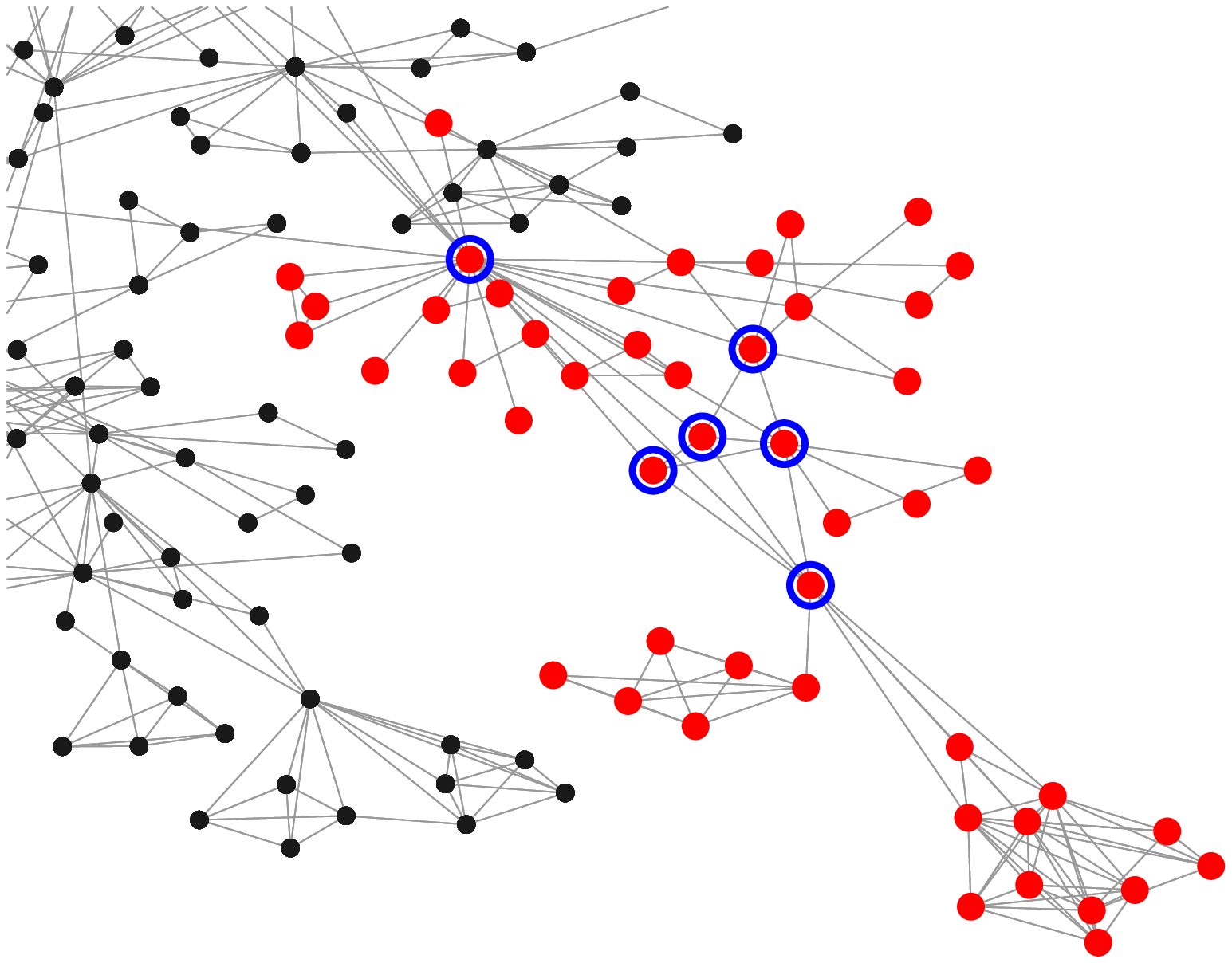}
\label{fig:subfigure6}}
\caption{\alg{SimpleLocal} results on a small graph for different $\delta$ values where $\delta$ controls the sparsity regularization and size of the output. The reference set is given by the nodes circled in blue, and the returned set $S^*$ is shown in red.}
\label{fig:Netscience}
\end{figure}



\subsection{MRI Scans}

\begin{figure}[tb]
\centering
\subfigure[ \fontsize{8}{9}\selectfont The target ventricle]{%
\includegraphics[width=0.35\linewidth]{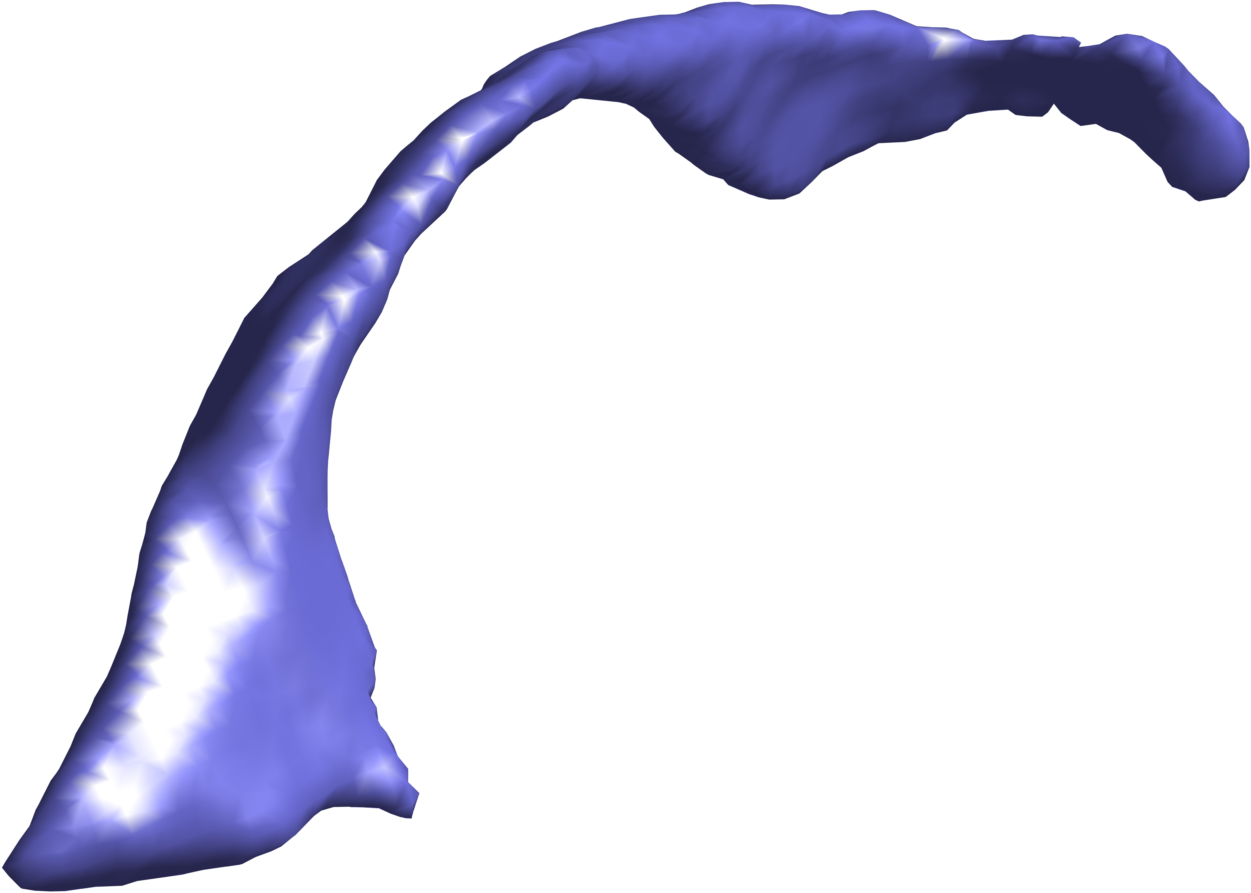}
}
\subfigure[\fontsize{8}{9}\selectfont {\alg{SimpleLocal} result}]{%
\includegraphics[width=0.35\linewidth]{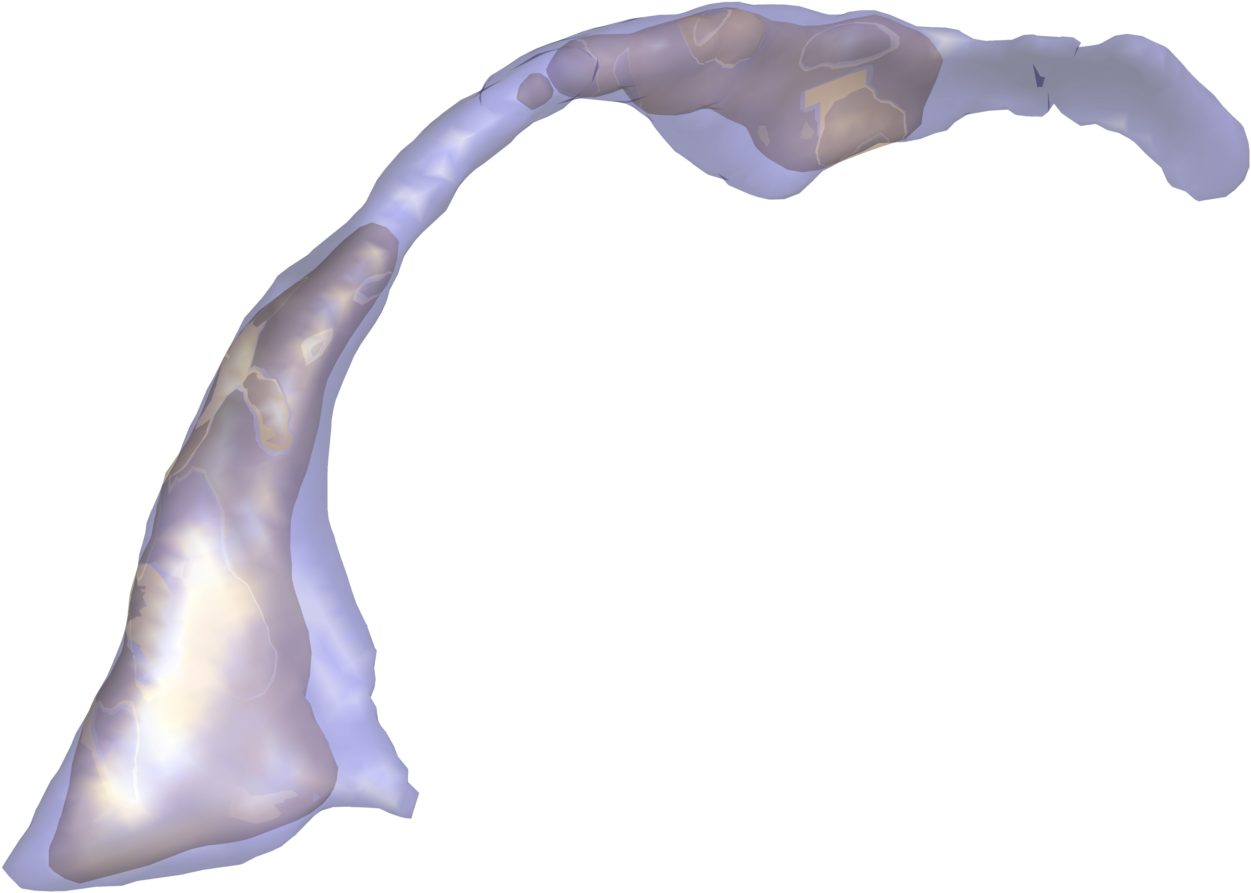}
}
\subfigure[\fontsize{8}{9}\selectfont Refined \alg{SimpleLocal}]{%
\includegraphics[width=0.35\linewidth]{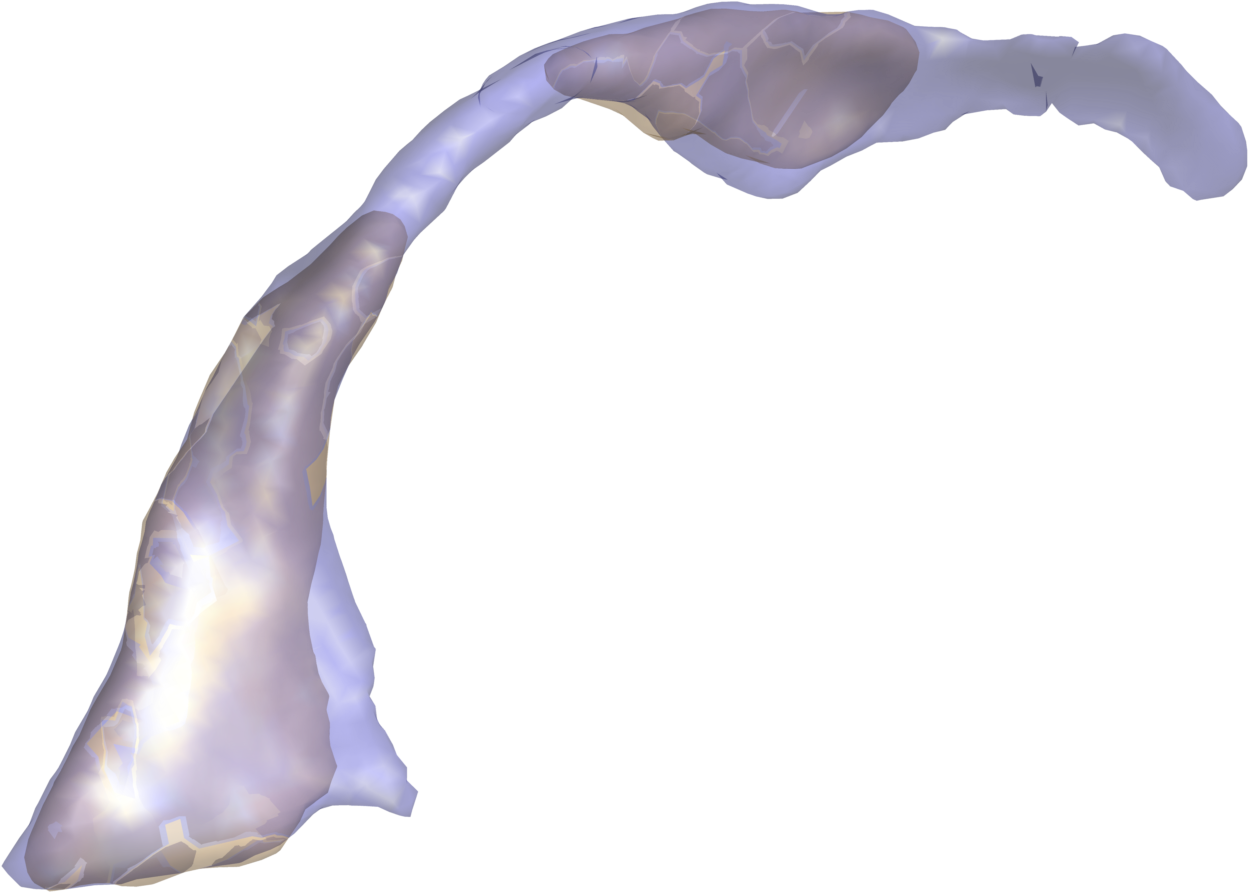}
}
\subfigure[\fontsize{8}{9}\selectfont  A spectral result]{%
\includegraphics[width=0.35\linewidth]{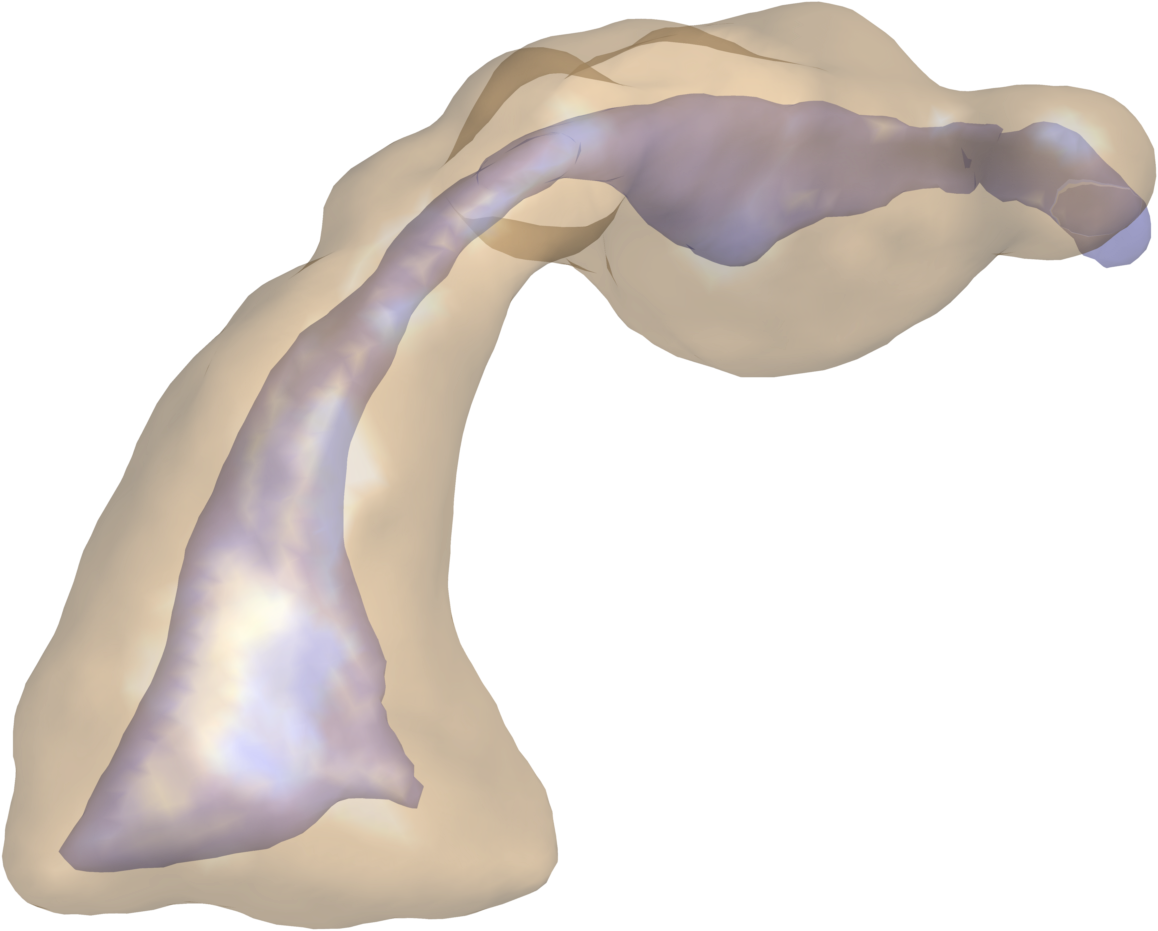}
}
\caption{Results from segmenting the left lateral ventricle in an 18 million voxel MRI scan on a graph with 467 million edges. The true set is always show in blue. Our new flow algorithms track the boundary closely but cannot find the bridge in the ventricle, whereas a spectral method returns a substantially larger region.}
\label{fig:mri}
\end{figure}

To demonstrate the scalability of our algorithm, we consider identifying a region in a 3d MRI scan. We obtained a labeled MRI scan from the MICCAI-2012 challenge with $256 \times 287 \times 256$ ($\approx$ 18 million) voxels~\cite{Marcus-2007-oasis}.  We formed a weighted graph based on adjacent voxel similarity (see supplement for details). The final graph contained around 467 million edges and 18 million voxels.

\begin{table}[tb]
\vspace{-\baselineskip}
\caption{Statistics on the true left lateral ventricle and the sets returned by the three methods \alg{SimpleLocal}, a refinement step, and spectral. The refined set gave the best conductance and highest accuracy overall.}
\label{tab:stats}
\fontsize{8}{9} \selectfont
\vskip .15in
 \begin{tabularx}{\linewidth}{@{}lXXp{1cm}lXX@{}}
 \toprule
 method & size & $\phi$ & volume explored & precision & recall & time (sec.) \\
 \midrule 
  True & 3965 & 0.129 & -- & -- & --  & -- \\ 
  \addlinespace
\alg{Sim.Loc.} & 2425 & 0.089 & 2463247 & 0.96 & 0.59 & 278.4 \\ 
+Refined & 2737 & 0.067 & 1845966 & 0.97 & 0.67 & +97.5 \\  \addlinespace
Spectral & 27918 & 0.094 & 5280988 & 0.14 & 0.99 & 9.6 \\ 
\bottomrule
 \end{tabularx}
\vspace*{-\baselineskip} 
\end{table}

The left lateral ventricle is a cavity in the interior of the brain shown in Figure~\ref{fig:mri}a. We use our \alg{SimpleLocal} method, a one-step \alg{3StageFlow} refinement procedure (typical of what might be done in practice), and a spectral method to identify this region from 75 randomly chosen seed voxels (Figures~\ref{fig:mri}b-d). (See the supplement for the details of the computations and parameter choices.) We present the statistics of the four sets in Table~\ref{tab:stats}. Overall, the flow method accurately tracks the true boundary of the region, although it is unable to complete an internal bridge within the region. The refinement step fills in the region slightly more. In comparison, the spectral method returns a much larger set that contains the entire ventricle, but completely misses the boundary. This mirrors the intuition from the introduction and results on this same spectral method in community detection, where it often finds \emph{larger, but imprecise} communities~\cite{Kloster-2014-hkrelax}. 

Note that the bridge of the ventricle is unlikely to be found by our method in this case. This happens because either flow-based set identified has a conductance value that is smaller (0.089 and 0.067) than the conductance of the entire region (0.129). Attempting to improve the conductance value will only shrink the identified region further (see the supplement for a few of these smaller, better conductance sets).  Another curious aspect of \alg{SimpleLocal}'s result set is that it is disconnected. The larger of the two regions actually has a smaller conductance value itself, but the method finds a disconnected set  because of the disconnected seeds. 
In terms of the runtime, the spectral method is faster than our sequence of max-flow problems. We discuss engineering details that could improve runtime in the supplement.

\section{Conclusions and Discussion}

We have given a new, simple, strongly-local algorithm for a commonly occurring problem that arises in semi-supervised learning, community detection on graphs, and image segmentation. This algorithm begins with a reference set that reflects a region of the graph known to be important and seeks a better conductance set nearby. Our method is heavily influenced by both the \alg{Improve} and \alg{LocalImprove} methods. In comparison with \alg{Improve}, our method is strongly-local and practically scalable (given a max-flow solver for the local graphs). In comparision with \alg{LocalImprove}, we have a significantly worse theoretical runtime because we solve a sequence of \emph{maximum flow} problems compared with their use of \emph{blocking flows}. However, our algorithm is simple to implement and can take advantage of many well-engineered maximum flow codes, such as Boykov and Kolmogorov's method that enables efficient modified flows~\cite{Boykov-2004-maxflow}. We also identified the implicit source of locality in the \alg{LocalImprove} method (Theorem~\ref{thm:l1}), which may enable even faster methods in the future. 

The new \alg{SimpleLocal} implementation enabled us to run experiments on a massive MRI scan with 467 million edges that would not have been possible or desirable in a weakly-local sense using traditional graph algorithms, because the output should be a set of roughly 4000 vertices out of 18 million.  Our work thus opens new possibilities in the use of maximum flows for machine learning. In particular, using a combination of spectral and flow methods will likely lead to improved results on many problems due to their complementary properties. Spectral methods can help quickly identify expanded, crude regions that the flow-based methods could contract to sharpen the boundaries. 

In future work, we plan to extend our contribution to approximate maximum-flow solutions. This would enable us to take advantage of recent innovations that produce approximate maximum-flows in nearly-linear time~\cite{Christiano-2011-max-flow,Lee-2013-max-flow,Sherman-2013-max-flow}---which would likely lead to a better theoretical runtime as well. Also, we wish to better understand the tradeoffs between spectral and flow methods using this new strongly-local computational primitive.

\section{Acknowledgments}

We'd like to acknowledge and thank several funding agencies for supporting our work. Gleich was supported by NSF awards IIS-1546488, Center for Science of Information STC, CCF-093937, CAREER CCF-1149756, and DARPA SIMPLEX. Veldt was supported by NSF award IIS-1546488. Mahoney would like to acknowledge the Army Research Office, the Defense Advanced Research Projects Agency, and the Department of Energy for providing partial support for this work.

\footnotesize
\bibliography{refs}
\bibliographystyle{icml2016}

\newpage
\appendix
\section{Extra results and proofs}

In this section, we include a few of the results we use that also appeared in other material -- or had very similar proofs in other material -- but restated in our notation for the reader's convenience.

\begin{lemma}
 If the minimum $s$-$t$ cut of $G'_R(\alpha,\delta)$ for $\delta \geq 0$ is less than $\alpha \vol(R),$ then $\phi(S) < \alpha$, where $S$ is the node set corresponding to the cut.

\end{lemma}
\begin{proof}
Recall that the min-cut objective can be stated as 
\begin{equation*}
\min_{S \subset V} \,\, \alpha \vol(R)+  \partial S - \alpha \vol(R\cap S) + (\alpha f(R) + \alpha \delta) \vol(\bar{R} \cap S).
\end{equation*}
If the objective is less than $\alpha \vol(R),$ then
\begin{align*}
& \partial S - \alpha \vol(R\cap S) + (\alpha f(R) + \alpha \delta) \vol(\bar{R} \cap S) < 0\\
& \implies \frac{\partial S}{\vol(R\cap S) - \eps \vol(\bar{R} \cap S)} < \alpha,
\end{align*}
where $\eps = f(R) + \delta.$
All we need to show then is that 
\[\vol(R\cap S) - \eps \vol(\bar{R} \cap S) \leq \min\{\vol(S), \vol(\bar{S})\}\]
and it will follow that
\[\phi(S) =  \frac{\partial S}{\min\{\vol(S), \vol(\bar{S})\}} < \alpha.\]

We first note that
\[\vol(R\cap S) - \eps \vol(\bar{R} \cap S) \leq \vol(R\cap S) \leq \vol(S).\]
Also,
\begin{align*}
&\vol(R\cap S) - \eps \vol(\bar{R} \cap S) \\
&\leq \vol(R\cap S) - f(R) \vol(\bar{R} \cap S)\\
&= \vol(R) - \vol(R \cap \bar{S}) - f(R)\vol(\bar{R}) + f(R)\vol(\bar{R}\cap \bar{S})\\
&\leq \vol(R) - f(R) \vol(\bar{R}) + f(R) \vol(\bar{R}\cap \bar{S})\\
&= f(R)\vol(\bar{R}\cap \bar{S})\\
&\leq \vol(\bar{S})
\end{align*}
so the result holds.
\end{proof}

Both assertions in the following theorem are novel results regarding our algorithm \alg{SimpleLocal}. They can be shown using the same proof techniques used in Lemma 2.2 of~\citet{andersen2008-improve}, with slight alterations to include the locality parameter $\delta$. 

\textbf{Theorem 4}
Given an initial reference set $R \subset V$ with $\vol(R) \leq \vol(\bar{R})$, \alg{SimpleLocal} returns a cut set $S^*$ such that 
\begin{enumerate}
\item if $C\subseteq R$, then $\phi(S^*) \leq \phi(C).$
\item For all sets of nodes $C$ such that
\[\frac{\vol(R\cap C)}{\vol(C)} \geq \frac{\vol(R)}{\vol(V)} + \gamma \frac{\vol(\bar{R})}{\vol(V)}\]
for some $\gamma > \delta$, we have
$\phi(S^*) \leq \frac{1}{(\gamma - \delta)} \phi(C).$
\end{enumerate}

\begin{proof}
We use the same proof outline as~\citet{andersen2008-improve}, and reproduce many of the same steps for the convenience of the reader.

The first assertion holds because if $C \subseteq R$, $\bar{\phi}_R(C) = \phi(C),$ so
\[\phi(S^*) \leq \bar{\phi}_R(S^*) \leq \bar{\phi}_R(C) = \phi(C),\]

where $\bar{\phi}_R$ is used to denote quotient score introduced in equation (7) of the paper. We refer to this as the \emph{modified} quotient score relative to $R$:

\[ \bar{\phi}_R(C) = \frac{\partial C}{\vol(R\cap C) - \eps \vol(\bar{R} \cap C)}.\]
To prove the second assertion we start by showing that $\bar{\phi}_R(C) \leq \frac{1}{(\gamma - \delta)} \phi(C)$, which is true if and only if \[\vol(C\cap R) - \eps \vol(C\cap \bar{R}) \geq (\gamma - \delta) \vol(C).\] To see this holds we apply the assumption made in the second assertion and simplify:
\begin{align*}
&\frac{\vol(R\cap C) - \eps \vol(C\cap \bar{R})}{\vol(C)} = \frac{\vol(C\cap R)}{\vol(C)} - \eps \frac{\vol(C\cap \bar{R})}{\vol(C)}\\
&\geq  \frac{\vol(R)}{\vol(V)} + \gamma \frac{\vol(\bar{R})}{\vol(V)} - (f(R)+\delta) \left(1 - \frac{\vol(R)}{\vol(V)} - \gamma \frac{\vol(\bar{R})}{\vol(V)}  \right)\\
&=\gamma \frac{\vol(\bar{R})}{\vol(V)}\left(1 + f(R)  \right) +  \frac{\vol(R)}{\vol(V)}\left( 1 - \frac{\vol(V)}{\vol(\bar{R})} + \frac{\vol(R)}{\vol(\bar{R})}\right) \\
&\hspace{1cm}- \delta   \left(1 - \frac{\vol(R)}{\vol(V)} - \gamma \frac{\vol(\bar{R})}{\vol(V)}  \right)\\
&= \gamma\cdot 1+ 0 - \delta \left(1 - \frac{\vol(R)}{\vol(V)} - \gamma \frac{\vol(\bar{R})}{\vol(V)}  \right)\\
&\geq \gamma - \delta.
\end{align*}
Since $S^*$ is the set that minimizes $\bar{\phi}_R(S)$, we have
\[\phi(S^*) \leq \bar{\phi}_R(S^*) \leq  \bar{\phi}_R(C) \leq \frac{1}{(\gamma - \delta)} \phi(C).\]
\end{proof}

\section{Empirical Runtime of SimpleLocal}
In terms of the runtime, the spectral method is substantially faster in practice than our sequence of max-flow problems. (See Table~\ref{tab:stats} in the main text.)
This arises due to a few factors. First, we are using a carefully engineered code for the spectral algorithm designed for speed. Second, we are using a general-purpose linear programming solver for the maximum-flow problems. Third, we are not exploiting any possible ``warm-start'' between independent flow solutions. We anticipate that a more careful implementation within our highly flexible three-stage framework would shrink the runtime gap considerably.

\section{Experiment parameters for the MRI problem}

We obtained a labeled MRI scan from the MICCAI-2012 challenge with $256 \times 287 \times 256$ voxels (around 18 million). (The MRI scans originated with the OASIS project, and labeled data was provided by Neuromorphometrics, Inc.~{\fontsize{8}{9}\selectfont \url{neuromorphometrics.com}} under an academic subscription.)  
We assembled a nearest neighbor graph on this image using all $26$ spatially adjacency voxels where each edge was weighted similar to~\citet{Shi2000-normalized-cuts}. We used the function $e^{-(\sqrt{I_i} - \sqrt{I_j})^2/0.05^2}$ where $\sqrt{I_i}$ is the scan intensity at voxel $i$. Subsequently, we threshholded the graph at a minimum weight of $0.1$ and scaled each edge weight to have minimum weight $1$ so that the volume of a set was an upper-bound on the number of edges contained.
The final graph was connected except for 35 voxels and contained 467 million edges. 

\paragraph{Seeding and SimpleLocal}
We picked $75$ random voxels in the true image, then used \alg{SimpleLocal} to refine the set $R$ consisting of these $75$ voxels and their immediate neighbors using a value of $\delta = 0.1$ to keep the computation local. The seed set is shown here in the supplement in Figure~\ref{fig:seed}. The resulting set is show in~\ref{fig:mri}(b).  

\paragraph{Refinement}
The output from \alg{SimpleLocal} can be further improved by growing the set by its neighborhood and varying $\delta$. We call this ``refinement" and used one step of refinement with $\delta = 0.5$.
The result is in Figure~\ref{fig:mri}(c).  

\paragraph{Spectral} We compare this against a highly-optimized strongly-local spectral method to minimize conductance using personalized PageRank vectors~\cite{andersen2006-local}, where the PageRank computation uses $\alpha = 0.99$. The spectral result is in the final subfigure Figure~\ref{fig:mri}(d). 

\paragraph{Parameter selection}
We picked parameters for the flow methods to ensure that the volume explored would be around 10 times the volume of the desired ventricle, and occasionally reduced the parameter $\delta$ if it seemed that the method was exploring too much or if the flow problems took too long. We picked the parameters for the spectral method until we found a set that meaningfully grew. Our particular technique attempts to avoid diffusing as much as possible and so we had to adjust the parameters to ensure that it moved beyond the seed set.

\subsection{Near optimality of Refined SimpleLocal}
We can use our \alg{SimpleLocal} and \alg{3StageFlow} primitives to attempt to identify the \emph{best} and \emph{largest} conductance set largely contained within the target ventricle. This is essentially the best result we could hope to achieve as the entire desired set has conductance larger than the set we identify. Thus, if we run a single iteration of \alg{3StageFlow} using the entire target set as $R$, $\alpha = 0.1291$ (the conductance of the target set), and $\delta = 15$, we will find a set that is almost exclusively contained within the target ventricle (Figure~\ref{fig:large-inside}). This choice of $\delta$ is guided by the intuition that we want the set to be \emph{almost exclusively} in the interior of the target, but small variations outside would be okay. The resulting solution set found has conductance $0.0621$ and 2527 vertices. The difference between the refined set we generated (Figure~\ref{fig:mri}(c) in the main text) and this set is slight. Their intersection is 2317 voxels. So there is a slightly better set that \alg{SimpleLocal} and the refinement procedure could have generated, but not by much.

\begin{figure}[tb]
 \includegraphics[width=\linewidth]{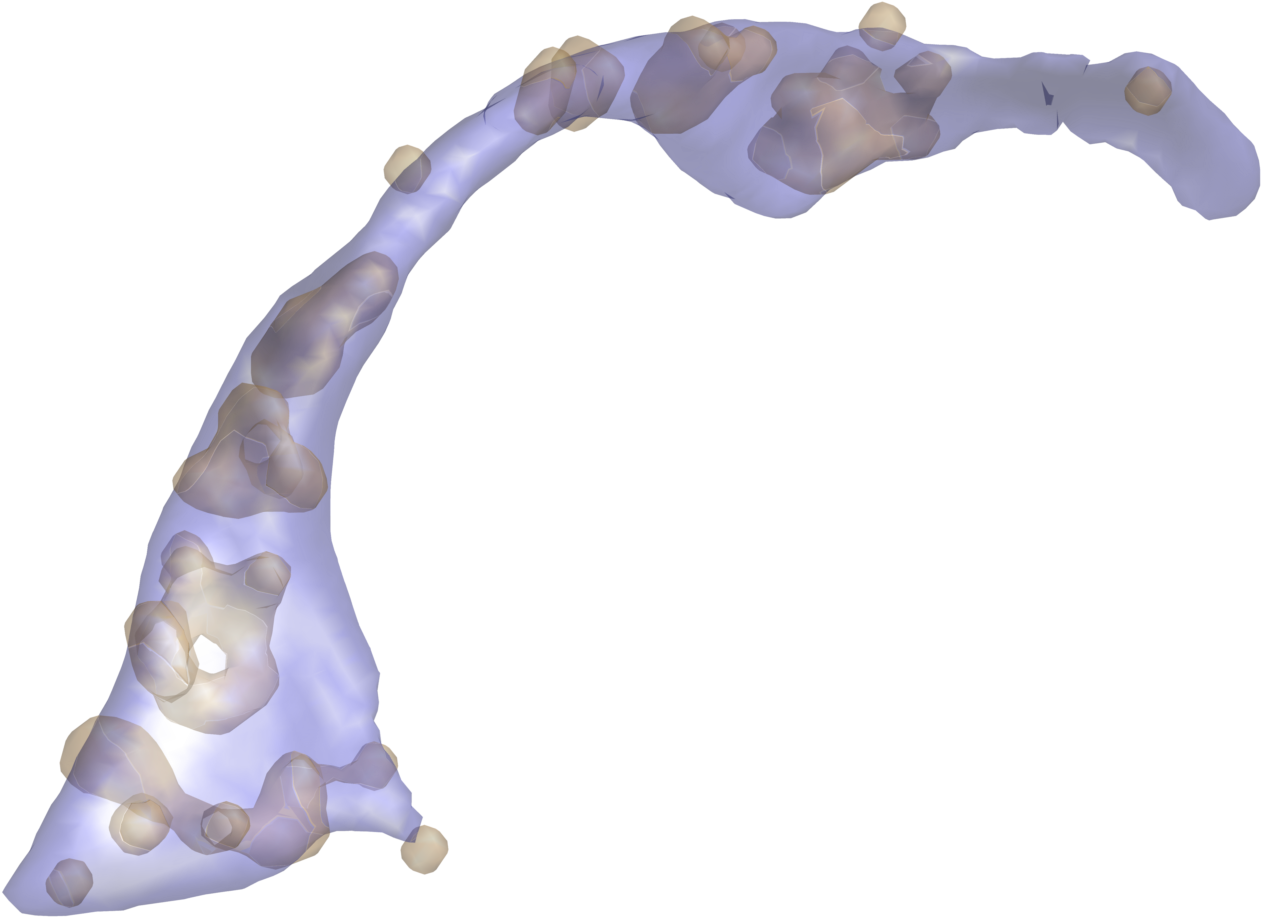}
 \caption{The seed set $R$ for the MRI segmentation. As in the main body figures, the true ventricle is shown in blue and the set in orange.}
 \label{fig:seed}
\end{figure}

\begin{figure}[tb]
 \includegraphics[width=\linewidth]{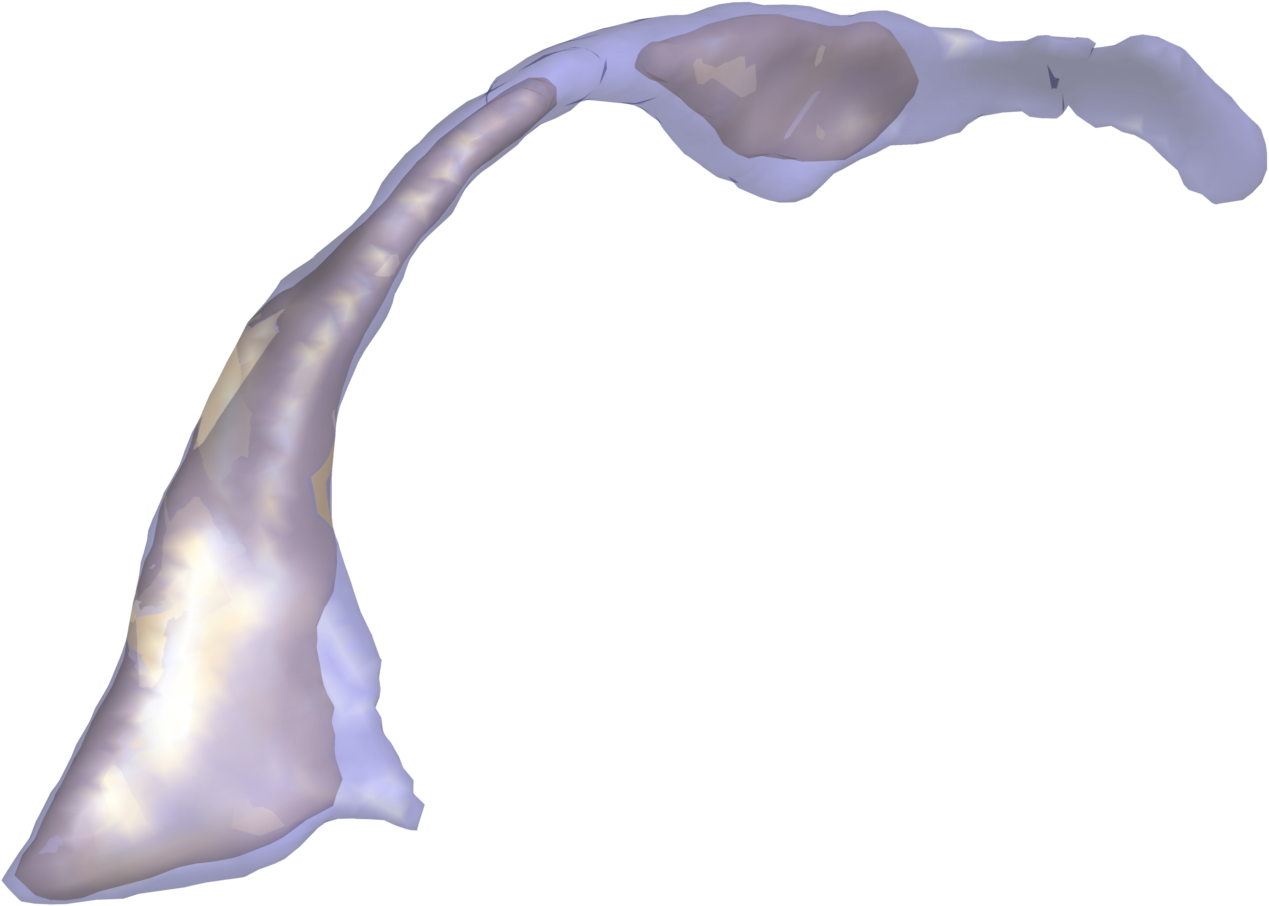}
 \caption{The largest set we identified with a small value of conductance inside the target ventricle. This is essentially the best set we could hope to identify from our flow techniques.}
 \label{fig:large-inside}
\end{figure}

\subsection{Other good sets}
We highlight a few other low-conductance sets we identified in the course of our experiments in Figure~\ref{fig:more1} and Figure~\ref{fig:more2}. In the first figure, we show another set available from the spectral method that makes a boundary error in the other direction and ends up too far inside the set. A closely related set in Figure~\ref{fig:more1} is, perhaps, the optimal set contained within the the target ventricle. It has the lowest conductance score of any set we ever computed. One challenge with using the flow-based methods such as \alg{SimpleLocal} is that they tend to quickly contract to very good, small sets. For instance, there is a set of $295$ vertices with very good conductance (Figure~\ref{fig:more2}). If the parameter $\delta$ is set too high, then this often causes the flow-based method to contract too much (e.g.~we over-regularize) and identify a very precise small set. This feature could be useful in some applications where the conductance measure is a very good proxy for the desired output.

\begin{figure}[tb]
 \subfigure[Optimized spectral]{\includegraphics[width=0.45\linewidth]{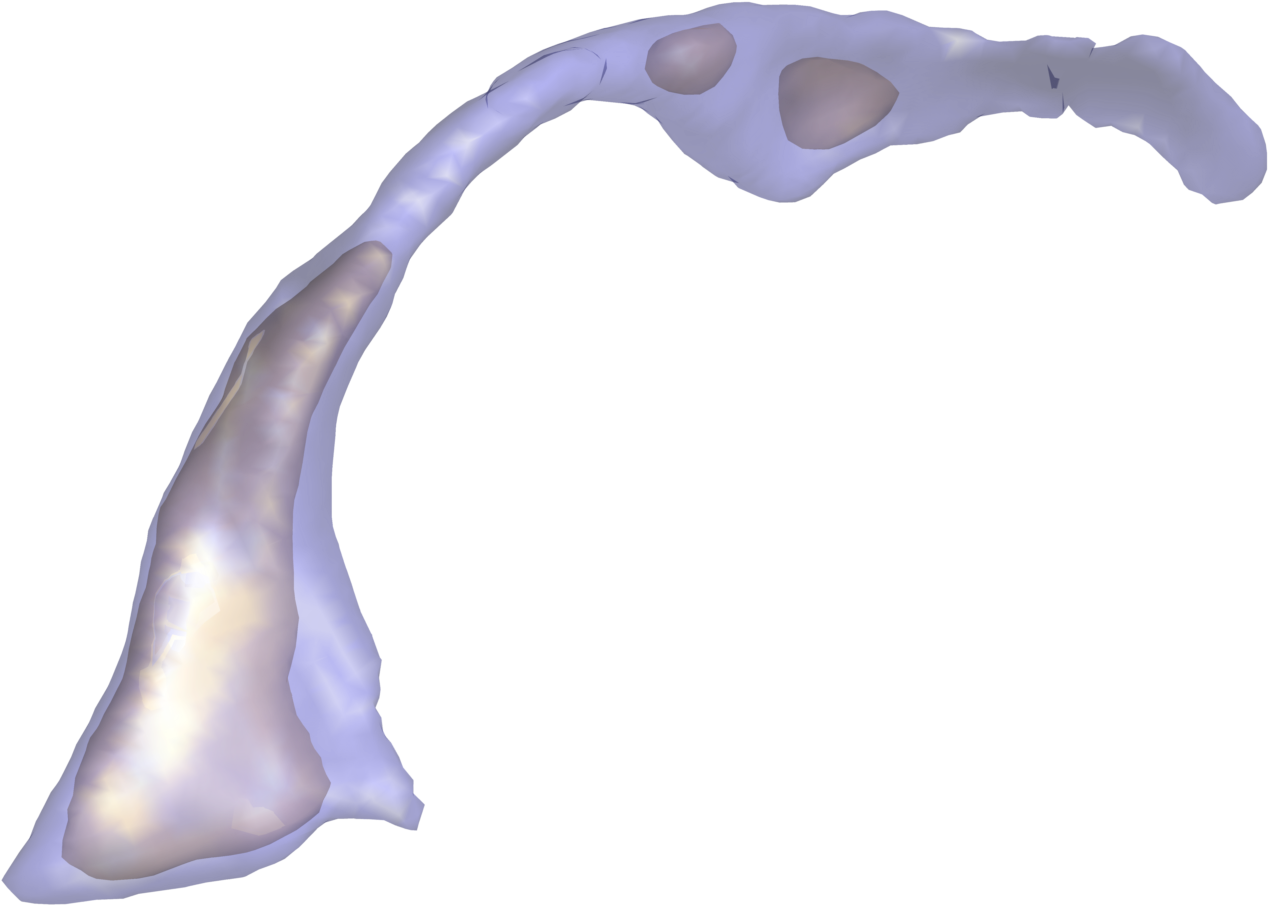}}
 \subfigure[Component of Refined]{\includegraphics[width=0.45\linewidth]{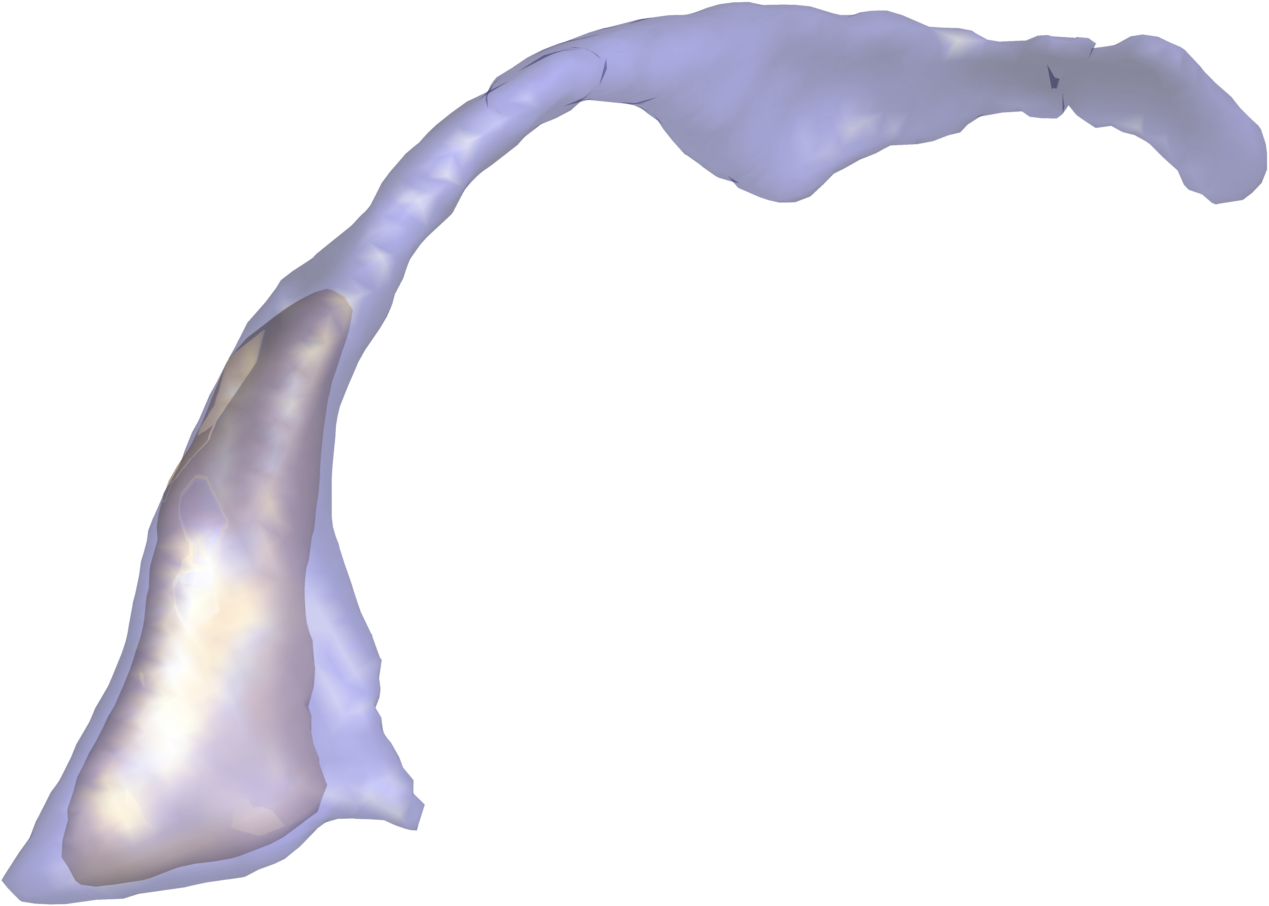}}
 \caption{At left, we have another set from spectral that identifies a low-conductance set nearly strictly inside. At right, we show the best subset of the disconnected region identified by \alg{SimpleLocal} and the refinement procedure. The spectral set has conductance $0.079$ and the \alg{SimpleLocal} component has conductance $0.0398$. Note that the spectral set does not hug the boundary nearly as closely as the results from the \alg{SimpleLocal} method in the main paper.}
 \label{fig:more1}
\end{figure}

\begin{figure}[t]
 \includegraphics[width=\linewidth]{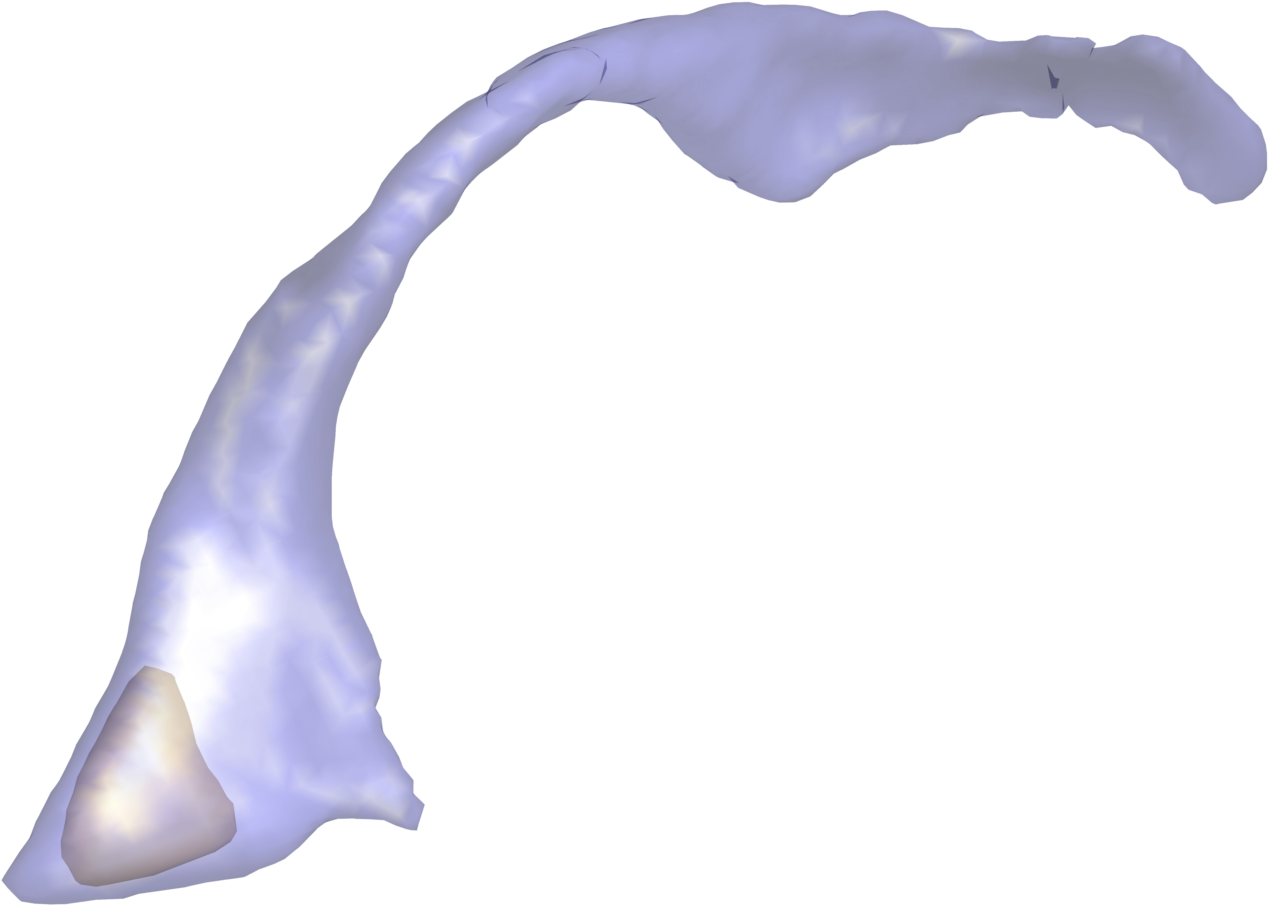}
 \caption{A tiny set of 295 vertices with conductance $0.048$ buried deep within the ventricle. This set often attracts the flow-based method if the value of $\delta$ is set too high.}
 \label{fig:more2}
\end{figure}

\end{document}